\documentclass[11pt]{article}
\usepackage[margin=1in]{geometry}
\usepackage[round,authoryear]{natbib}
\usepackage[utf8]{inputenc}
\usepackage[T1]{fontenc}
\usepackage[hidelinks,hypertexnames=false]{hyperref}
\usepackage{url}
\usepackage{booktabs}
\usepackage{amsfonts}
\usepackage{nicefrac}
\usepackage{microtype}
\usepackage{xcolor}
\usepackage{amsmath,amssymb,amsthm,mathtools}
\usepackage{algorithm}
\usepackage{algpseudocode}
\usepackage{threeparttable}
\usepackage{tabularx}
\usepackage{graphicx}
\graphicspath{{figures/}}
\usepackage{enumitem}
\usepackage{placeins}

\title{TS--Neyman: 
Posterior Sampling for Adaptive Stratified Estimation}
\author{%
Kosuke Morikawa, Mst Moushumi Pervin, and Jae Kwang Kim\\
Department of Statistics, Iowa State University
}
\date{}

\numberwithin{equation}{section}

\newtheorem{assumption}{Assumption}[section]
\newtheorem{proposition}[assumption]{Proposition}
\newtheorem{lemma}[assumption]{Lemma}
\newtheorem{theorem}[assumption]{Theorem}
\newtheorem{corollary}[assumption]{Corollary}
\newtheorem{remark}[assumption]{Remark}

\newcommand{\E}{\mathbb{E}}
\newcommand{\Prob}{\mathbb{P}}
\newcommand{\Var}{\operatorname{Var}}
\newcommand{\argmax}{\operatorname*{argmax}}
\newcommand{\1}{\mathbf{1}}
\newcommand{\wt}{\widetilde}
\newcommand{\wh}{\widehat}
\newcommand{\toas}{\xrightarrow{\mathrm{a.s.}}}
\newcommand{\rightsquig}{\rightsquigarrow}

\begin{document}

\maketitle

\begin{abstract}
Many model evaluation tasks reduce to estimating an average loss, error rate, or subgroup metric on a stratified pool when each label, human rating, or simulator call is costly. The precision-optimal Neyman allocation depends on within-stratum variances, which must be learned from the same observations used for estimation. We formulate this as a sequential allocation problem and use the exact one-step marginal variance reduction as the priority index. Replacing the unknown variances by independent inverse-$\chi^2$ posterior draws yields TS--Neyman, a Thompson-sampling rule that preserves the oracle marginal-gain structure while randomizing over variance uncertainty. For any fixed finite number of strata, we prove almost-sure convergence of the TS--Neyman allocation proportions to the Neyman target, asymptotic optimality of the variance proxy, and a central limit theorem for the resulting adaptive stratified estimator. In two five-stratum budget-scaling benchmarks, one bounded-loss and one binary model-error in the spirit of \citet{DaiGraduHarshaw2023ClipOGD}, TS-Neyman's relative efficiency stays within 5\% of the oracle on the bounded-loss population and within about 15\% on the binary benchmark.
 In an additional CivilComments real-data replay with confidence-based strata, it stays within about 8\% of the oracle and improves on equal allocation by roughly 7--14\% in MSE across budgets, while plug-in greedy and two-stage plug-in can degrade by over an order of magnitude under sparse pilots. Common-pilot warm-start and prior-sensitivity studies show this behavior is stable under working-model and working-prior misspecification.
\end{abstract}

\section{Introduction}
\label{sec:intro}

Many problems in machine learning require estimating a population-level performance quantity under a limited data-acquisition budget. A central example is budgeted model evaluation: a practitioner has a large pool of inputs, knows each input's class, domain, subgroup membership, or model-confidence bin, but must pay for labels, human ratings, expert review, or expensive simulator calls before observing the loss or error on that input. In this setting the pool is naturally stratified, and the central design question is not only how many evaluation labels to acquire, but \emph{where} to acquire them.

We use stratified-sampling notation because it gives an exact variance objective, but the same mathematical problem appears in ML evaluation and auditing. Each unit can be an input example, each stratum can be a class, subgroup, domain, or confidence bin, and the observed value $y_{hi}$ can be a loss, a binary error indicator, a calibration score, or a human-evaluation score. The objective is final estimator precision for an average performance metric, not cumulative reward during data collection.

When the objective is estimation rather than reward maximization, the optimal allocation depends on the heterogeneity of the strata. In classical stratified sampling, \emph{Neyman allocation} assigns samples in proportion to $N_h S_h$, where $N_h$ is the size of stratum $h$ and $S_h$ is the within-stratum standard deviation 
\citep{neyman1934two,cochran1977sampling}. This allocation minimizes the variance of the stratified mean estimator under a fixed budget, and design-based estimators such as the Horvitz--Thompson estimator provide the standard framework for finite-population inference under unequal sampling designs \citep{horvitz1952generalization}. However, while stratum sizes are typically known from a sampling frame, the within-stratum variances must usually be learned from the same observations used for estimation.

This creates a sequential decision problem. After $n_h$ samples are drawn from stratum $h$, the allocation-sensitive part of the stratified estimator's variance is proportional to a weighted sum of the terms $N_h^2 S_h^2 / n_h$, and adding one more observation to stratum $h$ reduces this term by $N_h^2 S_h^2 / [n_h(n_h+1)]$. This is a simple oracle priority index for sequential allocation, but $S_h^2$ is unknown and must be learned online. A purely plug-in rule that greedily samples the stratum with the largest \emph{estimated} marginal reduction can therefore overcommit to strata whose initial sample variances are spuriously large; conversely, excessive exploration of apparently low-variance strata wastes budget. The resulting exploration--exploitation trade-off is related to bandit and active-sampling problems, but the objective is different: minimize the final error of an estimator, not maximize cumulative reward \citep{EtoreJourdain2010Adaptive,carpentier2015adaptive}. 

We propose \textsc{TS--Neyman}, a Thompson sampling algorithm for adaptive Neyman allocation. The algorithm maintains a working posterior over each stratum variance. At each round, it draws a plausible variance from each posterior, computes the corresponding sampled marginal variance reduction, and allocates the next observation to the stratum with the largest sampled reduction. \textsc{TS--Neyman} is thus greedy with respect to a posterior-sampled world rather than a single plug-in estimate, favoring strata that appear important for variance reduction while continuing to explore those whose variances remain uncertain. Our use of posterior sampling follows the Thompson sampling principle for sequential decision-making under uncertainty \citep{thompson1933likelihood,russo2018tutorial}.

\paragraph{Related work.} Adaptive variance-based allocation has been studied in stratified Monte Carlo \citep{EtoreJourdain2010Adaptive,carpentier2015adaptive}, in survey sampling and multi-wave allocation \citep{salehi2010efficient,yang2025optimall}, and in responsive and adaptive survey design \citep{groves2006responsive,schouten2013optimizing}. The closest paper on adaptive Neyman allocation is CLIP-OGD \citep{DaiGraduHarshaw2023ClipOGD}, which studies a two-arm design-based potential-outcomes problem in which the learner chooses a treatment probability and uses a clipped online projected-gradient method to obtain sublinear Neyman regret. Our setting differs in four ways: we study general $H$-stratum mean estimation rather than two-arm treatment-effect estimation, the action at each round is a discrete stratum choice rather than a treatment probability, our main experiments use without-replacement stratified sampling from a fixed finite population, and the proposed algorithm is a Thompson-sampling rule over unknown stratum variances rather than an online-gradient method. A more detailed comparison appears in Section~~\ref{sec:relatedwork}.  

\paragraph{Contributions.}
\begin{itemize}
\item \textbf{Formulation.} We cast adaptive Neyman allocation as a sequential decision problem for stratified model evaluation and finite-population mean estimation under unknown within-stratum variances, with final estimator precision as the objective rather than cumulative reward.
\item \textbf{Algorithm.} We derive the exact one-step marginal variance-reduction index for sequential Neyman allocation and use it to construct \textsc{TS--Neyman}, a Thompson-sampling rule that incorporates posterior uncertainty in the unknown stratum variances and mitigates the premature-commitment behavior of plug-in greedy allocation.
\item \textbf{Theory.} For any fixed finite $H$, we prove that \textsc{TS--Neyman} learns the Neyman allocation proportions almost surely, that the resulting variance proxy is asymptotically optimal, and that the adaptive stratified estimator is consistent and asymptotically normal. The proof combines posterior concentration, a weighted chi-square Lyapunov function on the allocation simplex with a Robbins--Siegmund argument, and a martingale central limit theorem under the adaptive filtration.

\item \textbf{Experiments.} We benchmark against oracle, equal-allocation, two-stage plug-in, plug-in greedy, and an $H$-stratum projected-gradient comparator inspired by CLIP-OGD, on two controlled five-stratum populations chosen to expose sparse-pilot failure modes and on a CivilComments real-data replay with confidence-quintile strata.

\end{itemize}

The remainder of the paper is organized as follows. Section~\ref{sec:setup} introduces the stratified model-evaluation  setup and the oracle Neyman allocation problem. Section~\ref{sec:algorithm} presents \textsc{TS-Neyman}  together with the H-CLIP-OGD empirical baseline and a detailed relation to prior work. Section~\ref{sec:theory} studies the theoretical properties of \textsc{TS-Neyman} for fixed finite $H$. Section~\ref{sec:simulation} describes the numerical experiments and reproducibility details.

\section{Basic Setup}
\label{sec:setup}

The notation below is standard in stratified sampling, but throughout the
paper the finite population may be read as a stratified ML evaluation pool.
For example, $y_{hi}$ may be the loss or binary error of a fixed model on
example $i$ in stratum $h$, and sampling a unit means acquiring the label or
assessment needed to reveal that value.

Consider a finite population partitioned into $H$ strata with known sizes
$N_1, \ldots, N_H$ and total size $N = \sum_{h=1}^{H} N_h$. For stratum $h$,
let $\bar{Y}_h$ and $S_h^2$ denote the finite-population mean and variance,
and write $w_h := N_h S_h$ for the Neyman weight. If $n_h$ units are sampled
by simple random sampling without replacement within stratum $h$, the
stratified Horvitz--Thompson estimator of the population mean is
\begin{equation}
  \widehat{\mu} = \frac{1}{N} \sum_{h=1}^{H} \frac{N_h}{n_h}
  \sum_{i \in s_h} y_{hi},
  \qquad \sum_{h=1}^{H} n_h = n.
  \label{eq:ht}
\end{equation}


Its design variance is
\begin{equation}
  \mathrm{Var}(\widehat{\mu})
  = \frac{1}{N^2} \sum_{h=1}^{H} \frac{N_h^2}{n_h}
    \Big(1 - \frac{n_h}{N_h}\Big) S_h^2
  = \frac{1}{N^2} \sum_{h=1}^{H} \frac{w_h^2}{n_h}
    - \frac{1}{N^2} \sum_{h=1}^{H} N_h S_h^2.
  \label{eq:design-var}
\end{equation}
The second term in \eqref{eq:design-var} is allocation-independent, so the
exact finite-population design problem is equivalent to minimizing
\begin{equation}
  V(n; S) := \sum_{h=1}^{H} \frac{w_h^2}{n_h}
\label{eq:proxy}  
\end{equation}
subject to  $n_h \geq 1$ and 
$\sum_{h=1}^{H} n_h = n$. 
  The lower bound $n_h \geq 1$ matches the apportionment convention that every
item receives at least one allocation; under the $m_0$-pilot initialization
of Section~\ref{sec:algorithm} it is automatic. The corresponding continuous
optimum is the classical Neyman allocation
\begin{equation}
  p_h^\star = \frac{w_h}{\sum_{g=1}^{H} w_g},
  \qquad n_h^\star \propto w_h,
  \label{eq:neyman}
\end{equation}
going back to \citet{neyman1934two}; see also \citet{cochran1977sampling}.


The oracle allocation \eqref{eq:neyman} assumes known within-stratum standard
deviations. In practice these are learned from the same samples used for
estimation, which turns allocation into a sequential decision problem. 


For sequential design it is useful to write the exact one-step improvement.
If stratum $h$ currently has $n_h$ observations, then allocating one
additional sample there changes \eqref{eq:proxy} by
\begin{equation}
  \Delta_h(n_h; S_h^2)
  = \frac{w_h^2}{n_h} - \frac{w_h^2}{n_h + 1}
  = \frac{w_h^2}{n_h(n_h + 1)}.
  \label{eq:marginal-gain}
\end{equation}

A brief algebraic check shows that greedy assignment by $\Delta_h$ is not just locally improving but globally optimal. Using the telescoping identity $1/n_h = 1 - \sum_{k=1}^{n_h-1} 1/[k(k+1)]$,
\[
V(n; S) \;=\; \sum_{h=1}^H \frac{w_h^2}{n_h} \;=\; \sum_{h=1}^H w_h^2 \;-\; \sum_{h=1}^H \sum_{k=1}^{n_h - 1} \frac{w_h^2}{k(k+1)}.
\]
The first term is constant in the allocation, so minimizing $V(n; S)$ over integer allocations with $\sum_h n_h = t$ and $n_h \geq 1$ is equivalent to selecting the largest $t - H$ values from the family $\{w_h^2 / [k(k+1)] : h \in [H],\, k \geq 1\}$. At the moment stratum $h$ already holds $n_h = k$ observations, $\Delta_h(n_h; S_h^2) = w_h^2 / [k(k+1)]$, so greedy assignment by $\Delta_h$ picks exactly these largest values in order. The one-step rule \eqref{eq:marginal-gain} is therefore globally optimal at every cumulative budget --- a property the algorithm in Section~\ref{sec:algorithm} will inherit after replacing the unknown $S_h^2$ with posterior draws of the superpopulation analogue.

\section{Proposed Algorithm}
\label{sec:algorithm}

In Section~\ref{sec:setup}, the exact marginal-gain rule was written in terms of unknown within-stratum variability. To turn that rule into an adaptive algorithm, we use a Bayesian working model for the stratum variances and randomize through posterior sampling. We emphasize at the outset that this posterior is an \emph{algorithmic} working model rather than a substantive Bayesian assumption on the true data-generating process: the proofs in Section~\ref{sec:theory} use only the explicit inverse-$\chi^2$ representation of the variance draws together with finite second moments of the true within-stratum observations.

For the asymptotic analysis we adopt an i.i.d.\ within-stratum superpopulation view in which stratum $h$ yields an infinite stream $Y_{h1}, Y_{h2}, \ldots$ with mean $\mu_h$ and variance $\sigma_h^2$, where $\sigma_h^2$ is the superpopulation analogue of the finite-population $S_h^2$. The Neyman weight is now $w_h = N_h \sigma_h$, with target proportions $p_h = w_h / \sum_g w_g$. After observing data $\mathcal{D}_h(t) = \{Y_{h1}, \ldots, Y_{h, n_h(t)}\}$ in stratum $h$, we use the working model
\begin{equation}
Y_{hi} \mid \mu_h, \sigma_h^2 \sim \mathcal{N}(\mu_h, \sigma_h^2), \qquad \sigma_h^2 \sim \mathrm{Inv}\text{-}\chi^2(\nu_0, s_0^2), \qquad p(\mu_h) \propto 1.
\label{eq:prior}
\end{equation}
With $s_h^2(t)$ denoting the sample variance in stratum $h$, the induced working posterior is
\begin{equation}
\sigma_h^2 \mid \mathcal{D}_h(t) \;\sim\; \mathrm{Inv}\text{-}\chi^2\!\left( \nu_0 + n_h(t) - 1, \; \frac{\nu_0 s_0^2 + (n_h(t) - 1) s_h^2(t)}{\nu_0 + n_h(t) - 1} \right).
\label{eq:posterior}
\end{equation}

\subsection{Thompson sampling for adaptive Neyman allocation}
\label{sec:tsneyman}

Based on the working posterior in \eqref{eq:posterior}, \textsc{TS--Neyman} evaluates the same marginal-gain structure as the oracle index of Section~\ref{sec:setup}, but with randomized variance draws. At step $t+1$, \textsc{TS--Neyman} draws $\widetilde\sigma_h^2(t)$ from \eqref{eq:posterior} and chooses the stratum with the largest sampled marginal gain
\begin{equation}
\widetilde\Delta_h(t) \;=\; \frac{N_h^2 \, \widetilde\sigma_h^2(t)}{n_h(t)\,\{n_h(t)+1\}}.
\label{eq:tsindex}
\end{equation}
The full procedure is summarized in Algorithm~\ref{alg:ts-neyman}.  
Relative to a plug-in greedy rule, \textsc{TS--Neyman} preserves uncertainty about each stratum's variance, which induces exploration automatically when a stratum has been sampled only a few times. The finite-population feasible set in Algorithm~\ref{alg:ts-neyman} is included only to avoid impossible allocations when a stratum is exhausted; the experiments choose budgets for which this cap is not binding for the main comparisons.

\begin{remark}[Computation and unequal costs]\label{rem:cost}
Each round requires $H$ posterior draws and $H$ index evaluations, so the per-round cost is $\mathcal{O}(H)$. If one sample from stratum $h$ costs $c_h > 0$, the marginal-gain calculation in \eqref{eq:tsindex} is replaced by $N_h^2 \, \widetilde\sigma_h^2(t) / [c_h \, n_h(t)\{n_h(t)+1\}]$, whose continuous target is $n_h^\star \propto N_h \sigma_h / \sqrt{c_h}$.
\end{remark}
\begin{remark}[Connection to apportionment]\label{rem:apportionment}
The deterministic marginal-gain rule whose optimality was verified at the end of Section~\ref{sec:setup} is, with $w_h = N_h \sigma_h$, exactly the Huntington--Hill apportionment of $t$ items among $H$ classes with weights $w_h$ \citep{huntington1928apportionment,balinski1982fair}. \textsc{TS--Neyman} can therefore be read as a posterior-sampling lift of Huntington--Hill apportionment to the unknown-variance setting, where at each step $w_h^2 = N_h^2 \sigma_h^2$ is replaced by an independent draw $N_h^2 \widetilde\sigma_h^2(t)$ from the working posterior.
\end{remark}

\begin{algorithm}[t]
\caption{TS-Neyman}
\label{alg:ts-neyman}
\begin{algorithmic}[1]
\Require Stratum sizes $\{N_h\}_{h=1}^H$, total budget $n$, initial sample size $m_0\ge 2$ per stratum, hyperparameters $(\nu_0,s_0^2)$.
\Require In finite-population runs, feasible set $\mathcal E_t=\{h:n_h(t)<N_h\}$; in the superpopulation theory, $\mathcal E_t=\{1,\dots,H\}$.
\For{$h=1,\dots,H$}
    \State Collect $m_0$ observations from stratum $h$ and initialize $n_h\gets m_0$.
\EndFor
\For{$t=Hm_0,\dots,n-1$}
    \State Set $\mathcal E_t=\{h:n_h(t)<N_h\}$ in finite-population runs and $\mathcal E_t=\{1,\dots,H\}$ in the superpopulation theory.
    \For{$h\in\mathcal E_t$}
        \State Draw $\wt\sigma_h^2(t)$ from \eqref{eq:posterior}.
        \State Compute $\wt\Delta_h(t)=N_h^2\wt\sigma_h^2(t)/[n_h(t)\{n_h(t)+1\}]$.
    \EndFor
    \State Choose $A_{t+1}\in\argmax_{h\in\mathcal E_t} \wt\Delta_h(t)$, breaking ties uniformly at random.
    \State Sample one additional unit from stratum $A_{t+1}$ and update that stratum's posterior.
\EndFor
\State \Return Final allocation $(n_1,\dots,n_H)$ and the stratified sample.
\end{algorithmic}
\end{algorithm}

\subsection{An H-stratum projected-gradient baseline}
\label{sec:hclipogd}

To strengthen the empirical comparison to the closest online-allocation literature, we also implement $H$-CLIP-OGD, a projected-gradient baseline inspired by CLIP-OGD~\citep{DaiGraduHarshaw2023ClipOGD}. Whereas CLIP-OGD is a two-arm sequential-experiment method whose decision variable is a single treatment probability, our setting requires a simplex-valued design variable on $\Delta_H$. We use a continuous proxy of the Neyman objective, take a plug-in gradient step on $\Delta_H$, and apply an $\varepsilon_t$-mixture parametrization to guarantee persistent exploration. The full construction (objective, gradient, step size, projection, and clipping schedule) is given in Appendix~\ref{app:hclipogd}. We use $H$-CLIP-OGD purely as an empirical comparator: no regret theorem is claimed, and all formal results in Section~\ref{sec:theory} concern \textsc{TS--Neyman} alone.

\subsection{Relation to prior work}
\label{sec:relatedwork}

Our work sits between three strands of literature, and is complementary rather than competing with each.

\textbf{Exact known-variance Neyman allocation.} Classical Neyman allocation \citep{neyman1934two,cochran1977sampling} has dynamic-programming and exact-algorithmic treatments 
\citep{Kadane2005OptimalDynamic,Wright2017ExactOptimal,Wright2020GeneralExact}, all of which assume known within-stratum variances. We treat the unknown-variance regime, where the marginal-gain index must be evaluated under uncertainty.

\textbf{Adaptive and Bayesian variance-based allocation.} Adaptive variance-based allocation has appeared in stratified Monte Carlo \citep{EtoreJourdain2010Adaptive,carpentier2015adaptive}, Bayesian finite-population design \citep{Zacks1970BayesianDesign}, stratified evaluation \citep{YuZhaiSra2019NearOptimal}, and two-stage or batch adaptive experimentation \citep{HahnHiranoKarlan2011,BlackwellPashleyValentino2023,Zhao2023AdaptiveNeyman}. Most of these rules are either non-randomized (two-stage plug-in, deterministic schedules) or operate at the batch level rather than one-by-one. \textsc{TS--Neyman} instead preserves the \emph{exact} one-by-one marginal-gain structure of the stratified-mean variance objective and randomizes through posterior sampling.

\textbf{CLIP-OGD and online-gradient methods.} The most directly comparable paper is CLIP-OGD \citep{DaiGraduHarshaw2023ClipOGD}, which studies adaptive Neyman allocation in two-arm sequential experiments under a design-based potential-outcomes framework, with treatment probability as the decision variable and a clipped online projected-gradient method that achieves sublinear expected Neyman regret. Our setting is genuinely different: $H$-stratum finite-population mean estimation rather than two-arm treatment-effect estimation, a discrete stratum choice rather than a continuous treatment probability, and a Thompson-sampling rule built from posterior draws of the unknown variances rather than an online-gradient method. Our contribution is therefore complementary to both the exact-allocation and online-gradient literatures: we develop a posterior-sampling rule for the unknown-variance case, prove fixed-$H$ allocation consistency (Theorem~\ref{thm:consistency}) and an estimator CLT (Corollary~\ref{cor:clt}), and instantiate $H$-CLIP-OGD as a fair empirical comparator on the $H$-simplex.

\section{Theoretical Properties}
\label{sec:theory}

We now give the asymptotic theory for TS--Neyman with any fixed finite number of strata.
The proof is written for the flat $H$-stratum algorithm in Algorithm~\ref{alg:ts-neyman}, not for a merged two-stratum surrogate.
The key technical point is that, once the posterior variance draws have concentrated, the sampled marginal-gain rule selects strata with the smallest relative allocation ratio $x_h(t)/p_h$, where $x_h(t)=n_h(t)/t$ and $p_h$ is the Neyman target.
A weighted chi-square Lyapunov function on the allocation simplex then gives simultaneous convergence of all coordinates.

\begin{assumption}[Fixed-$H$ iid TS--Neyman regime]
\label{ass:fixed-h}
The following conditions hold.
\begin{enumerate}[leftmargin=1.5em]
\item $H\ge2$ is fixed and finite.  Each stratum produces an independent iid stream with
$\E[Y_{h1}]=\mu_h$ and $\Var(Y_{h1})=\sigma_h^2\in(0,\infty)$.
\item The algorithm is initialized with $m_0\ge 2$ samples per stratum.
\item At time $t$, TS--Neyman draws $\wt\sigma_h^2(t)$ from the working posterior \eqref{eq:posterior} independently across strata conditional on the past, and selects the stratum with largest sampled index \eqref{eq:tsindex}, breaking ties uniformly at random.
\end{enumerate}
Let
\begin{equation}
w_h=N_h\sigma_h,
\qquad
W=\sum_{g=1}^H w_g,
\qquad
p_h=\frac{w_h}{W},
\qquad h=1,\dots,H. 
\label{eq:fixedh-target}
\end{equation}

\end{assumption}

The following theorem shows that the empirical allocation of \textsc{TS--Neyman} converges almost surely to the oracle Neyman target $(p_1, \ldots, p_H)$, even though the unknown variances $\sigma_h^2$ are learned only through randomized posterior draws rather than supplied as known inputs. 
\begin{theorem}[Allocation consistency for fixed $H$]
\label{thm:consistency}
Under Assumption~\ref{ass:fixed-h},
\[
\frac{n_h(t)}{t}\toas p_h,
\qquad h=1,\dots,H.
\]
\end{theorem}

The proof, given in Appendix~\ref{app:proofs}, has four ingredients.
First, every stratum is sampled infinitely often; otherwise its inverse-$\chi^2$ working posterior remains nondegenerate while the indices of increasingly sampled strata vanish, forcing it to be selected again.
Second, along every stratum, the working posterior for $\sigma_h^2$ concentrates on the true variance.
Third, if the allocation vector is away from the Neyman target, then at least one relative ratio $x_h(t)/p_h$ is bounded below one; after posterior concentration the TS index ranks such under-allocated strata above sufficiently over-allocated ones.
Finally, the weighted chi-square Lyapunov function
$Q_t=\sum_{h=1}^H p_h\left(x_h(t)/p_h-1\right)^2$ has negative drift whenever it is bounded away from zero, and a Robbins--Siegmund argument yields $Q_t\to0$ almost surely.

The following theorem shows that the variance proxy realized by \textsc{TS--Neyman} is asymptotically as small as any oracle integer allocation could make it, with the ratio of realized to optimal variance converging to one almost surely. 
\begin{theorem} 
\label{thm:efficiency}
Under Assumption~\ref{ass:fixed-h}, let
\[
V(\mathbf n)=\sum_{h=1}^H \frac{N_h^2\sigma_h^2}{n_h}
=\sum_{h=1}^H \frac{w_h^2}{n_h},
\]
and let $V^\star(t)$ denote the minimum of $V(\mathbf m)$ over integer allocations with $\sum_h m_h=t$ and $m_h\ge 1$.
Then
\[
\frac{V\bigl(\mathbf n(t)\bigr)}{V^\star(t)}\toas 1.
\]
\end{theorem}

The proof, given in Appendix~\ref{app:proofs}, has two ingredients. By Theorem~\ref{thm:consistency}, $n_h(t) = t p_h + o(t)$ almost surely, which yields $V(n(t)) = W^2/t + o(1/t)$ by direct substitution and the identity $w_h^2/(t p_h) = W w_h / t$. A Cauchy--Schwarz lower bound, together with a matching upper bound from any integer rounding of the continuous Neyman allocation, shows $V^\star(t) = W^2/t + O(t^{-2})$. The ratio therefore converges to one almost surely.

\begin{corollary}[Adaptive estimator consistency and CLT for fixed $H$]
\label{cor:clt}
Under Assumption~\ref{ass:fixed-h}, the adaptive stratified estimator
\begin{equation}
\wh\mu_t
=
\frac{1}{N}\sum_{h=1}^H N_h \bar Y_h(t)
\label{eq:mu-t}
\end{equation}
satisfies
\[
\wh\mu_t\toas \mu:=\frac{1}{N}\sum_{h=1}^H N_h\mu_h,
\]
and
\[
\sqrt{t}\,\bigl(\wh\mu_t-\mu\bigr)
\rightsquig
\mathcal N\!\left(
0,\;
\frac{1}{N^2}\sum_{h=1}^H \frac{N_h^2\sigma_h^2}{p_h}
\right)
=
\mathcal N\!\left(
0,\;
\frac{1}{N^2}\left(\sum_{h=1}^H N_h\sigma_h\right)^2
\right).
\]
\end{corollary}
This corollary shows that the adaptive estimator $\widehat{\mu}_t$ achieves the same $\sqrt{t}$-rate and asymptotic variance as an oracle Neyman allocation, so \textsc{TS--Neyman} pays no asymptotic cost for learning the within-stratum variances online. 

The proof, given in Appendix~\ref{app:proofs}, proceeds in two stages. Consistency follows from a martingale strong law applied to $M_h^Y(t) = R_h(t) - \mu_h n_h(t)$ together with $n_h(t) \to \infty$ from Theorem~\ref{thm:consistency}; this yields $\bar{Y}_h(t) \to \mu_h$ for each stratum. For the CLT, the predictable covariation $(1/t) \langle M^Y \rangle(t)$ converges almost surely to $\mathrm{diag}(\sigma_h^2 p_h)$ once the allocation has stabilized, a conditional Lindeberg condition holds under finite second moments, and the martingale CLT~\citep{HallHeyde1980} delivers the joint normal limit; Slutsky composition with $t/n_h(t) \to 1/p_h$ then produces the displayed variance $W^2/N^2 = (\sum_h N_h \sigma_h)^2 / N^2$.

\section{Numerical Experiments}
\label{sec:simulation}

We evaluate \textsc{TS--Neyman} on two controlled five-stratum budget-scaling benchmarks and one offline real-data replay. In all experiments, sampling is without replacement within strata, the final estimator is the stratified Horvitz--Thompson estimator, and performance is summarized by Monte Carlo bias, RMSE, and relative efficiency
\[
\mathrm{RE}(\mathcal A)=\frac{\mathrm{RMSE}(\mathcal A)^2}{\mathrm{RMSE}(\mathrm{Oracle})^2},
\]
where smaller is better and $\mathrm{RE}=1$ matches the oracle Neyman rule.

The two controlled benchmarks share the same five stratum sizes and budget grid,
\[
(N_1,N_2,N_3,N_4,N_5)=(50000,30000,15000,4000,1000),
\qquad
n\in\{100,200,400,600,1000,2000\}.
\]
They play complementary roles. The \emph{bounded-loss benchmark} is a controlled continuous-outcome population that isolates the unknown-variance / sparse-pilot effect in a clean setting. The \emph{binary model-error benchmark} is closer to the ML evaluation use case, in the budget-scaling style of \citet{DaiGraduHarshaw2023ClipOGD}: each sampled unit can be read as a held-out example, each stratum as a confidence or difficulty bin, and the observation as a binary error indicator. For the bounded-loss benchmark, the stratum means and scales are
\[
(\mu_1,\ldots,\mu_5)=(0.20,0.25,0.35,0.50,0.65),
\qquad
(\sigma_1,\ldots,\sigma_5)=(0.08,0.12,0.18,0.25,0.32),
\]
whereas in the binary benchmark each unit has
\[
Y=\mathbf 1\{\text{model error}\},
\qquad
(q_1,\ldots,q_5)=(0.02,0.05,0.10,0.20,0.35),
\]
so the strata become progressively rarer and noisier error bins.

\begin{figure*}[t]
\centering
\includegraphics[width=0.98\textwidth]{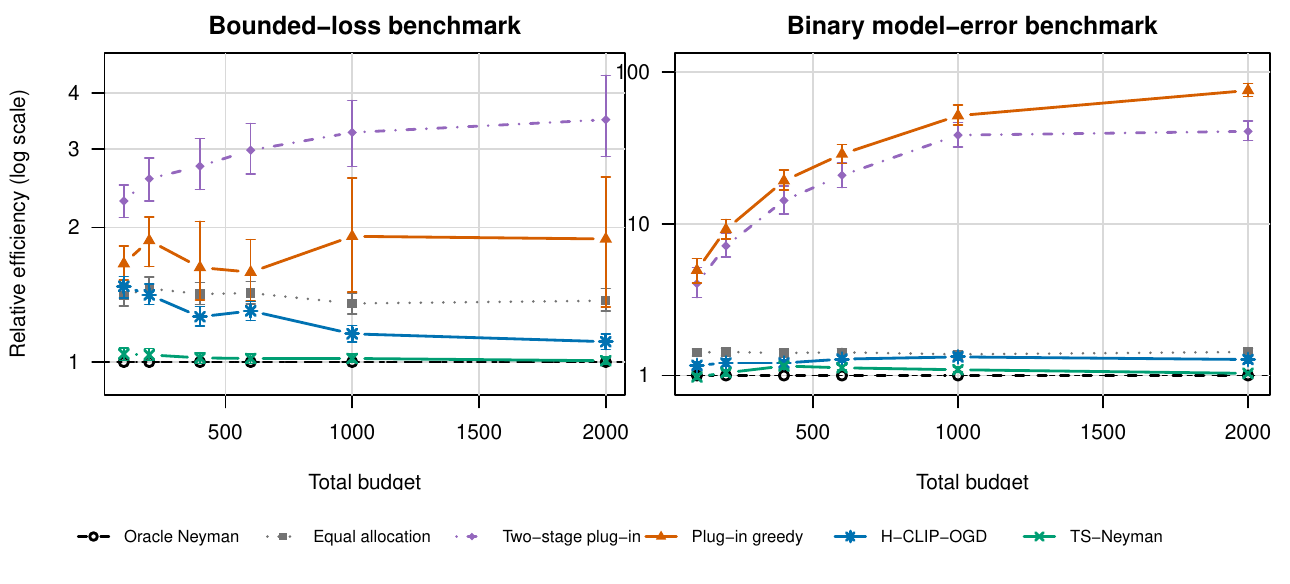}
\caption{Budget-scaling comparison for the two main controlled benchmarks. Left: bounded-loss benchmark. Right: binary model-error benchmark in the style of \protect\citet{DaiGraduHarshaw2023ClipOGD}. Both panels report relative efficiency on a log scale; the right panel uses ticks at $1$, $10$, and $100$ because the plug-in baselines deteriorate much more severely in the binary setting.}
\label{fig:budget_pair_main}
\end{figure*}

Figure~\ref{fig:budget_pair_main} compares all methods on both benchmarks. In the bounded-loss panel, TS--Neyman attains RE
\[
1.031,\;1.022,\;1.014,\;1.017,\;1.021,\;1.012,
\]
so it remains within about 5\% of the oracle across the full grid. Over the same budgets, equal allocation ranges from RE $1.358$ to $1.422$, H-CLIP-OGD from $1.119$ to $1.447$, plug-in greedy from $1.688$ to $2.368$, and two-stage plug-in from $2.234$ to $3.896$. In the binary model-error panel, TS--Neyman attains RE
\[
0.988,\;1.066,\;1.125,\;1.141,\;1.080,\;1.052,
\]
so it stays within about 15\% of the oracle even in the sparse-binary regime where plug-in variance estimates are especially unstable. There equal allocation remains between RE $1.358$ and $1.488$, H-CLIP-OGD between $1.170$ and $1.322$, but two-stage plug-in explodes from RE $3.788$ to $37.933$ and plug-in greedy from $4.659$ to $78.837$. Table~\ref{tab:budget_pair_main} summarizes the RE ranges across the two controlled benchmarks. Together, Figure~\ref{fig:budget_pair_main} and Table~\ref{tab:budget_pair_main} show that TS--Neyman tracks the oracle well, while deterministic plug-in baselines become brittle when only a small pilot is available.

\begin{table}[h]
\caption{Range of relative efficiencies across budgets in the two main benchmarks. Oracle Neyman is normalized to 1 and omitted.}
\label{tab:budget_pair_main}
\centering
\begin{tabular}{lcccc}
\toprule
Benchmark & Plug-in greedy & H-CLIP-OGD & TS--Neyman & Two-stage plug-in \\
\midrule
Bounded loss       & 1.688--2.368 & 1.119--1.447 & 1.012--1.031 & 2.234--3.896 \\
Binary model error & 4.659--78.837 & 1.170--1.322 & 0.988--1.141 & 3.788--37.933 \\
\bottomrule
\end{tabular}
\end{table}

\paragraph{CivilComments real-data replay.} To check that this behavior is not confined to synthetic populations, we also run an offline replay on CivilComments-WILDS \citep{koh2021wilds}, derived from the Civil Comments / Jigsaw unintended-bias data \citep{borkan2019nuanced}. We train a fixed text classifier on the provided training split, choose regularization on the validation split, and treat the test split as a finite evaluation pool. For each comment $i$, we define $Y_i=\mathbf 1\{\text{the fixed classifier misclassifies comment }i\}$ and stratify the test pool by classifier-confidence quintiles. The empirical mean error falls from $0.312$ in the least-confident quintile to $0.021$ in the most-confident quintile. Figure~\ref{fig:civilcomments_main} summarizes the replay diagnostics. Panel~(a) shows this monotone error gradient, and panel~(b) shows that TS--Neyman learns nearly the same allocation profile as the oracle: at budget $n=2000$, the oracle allocates $703$ labels to $Q_1$ and TS--Neyman allocates $695$, whereas equal allocation uses $408$ and plug-in greedy overshoots to $1506$. Panels~(c) and~(d) report RMSE and relative efficiency over the same budget grid using $500$ Monte Carlo replays. TS--Neyman attains RE $0.900, 1.006, 1.004, 1.021, 1.045, 1.076$, remaining within about $8\%$ of the oracle throughout. Equal allocation ranges from RE $1.011$ to $1.212$, while two-stage plug-in ranges from $1.316$ to $15.731$ and plug-in greedy from $1.790$ to $33.604$. Thus the same pattern persists on a real evaluation pool: posterior sampling stays close to the oracle while avoiding severe early overcommitment. Appendix Table~\ref{tab:civilcomments_re_appendix} reports the full per-budget RE summary.

\begin{figure*}[t]
\centering
\includegraphics[width=0.98\textwidth]{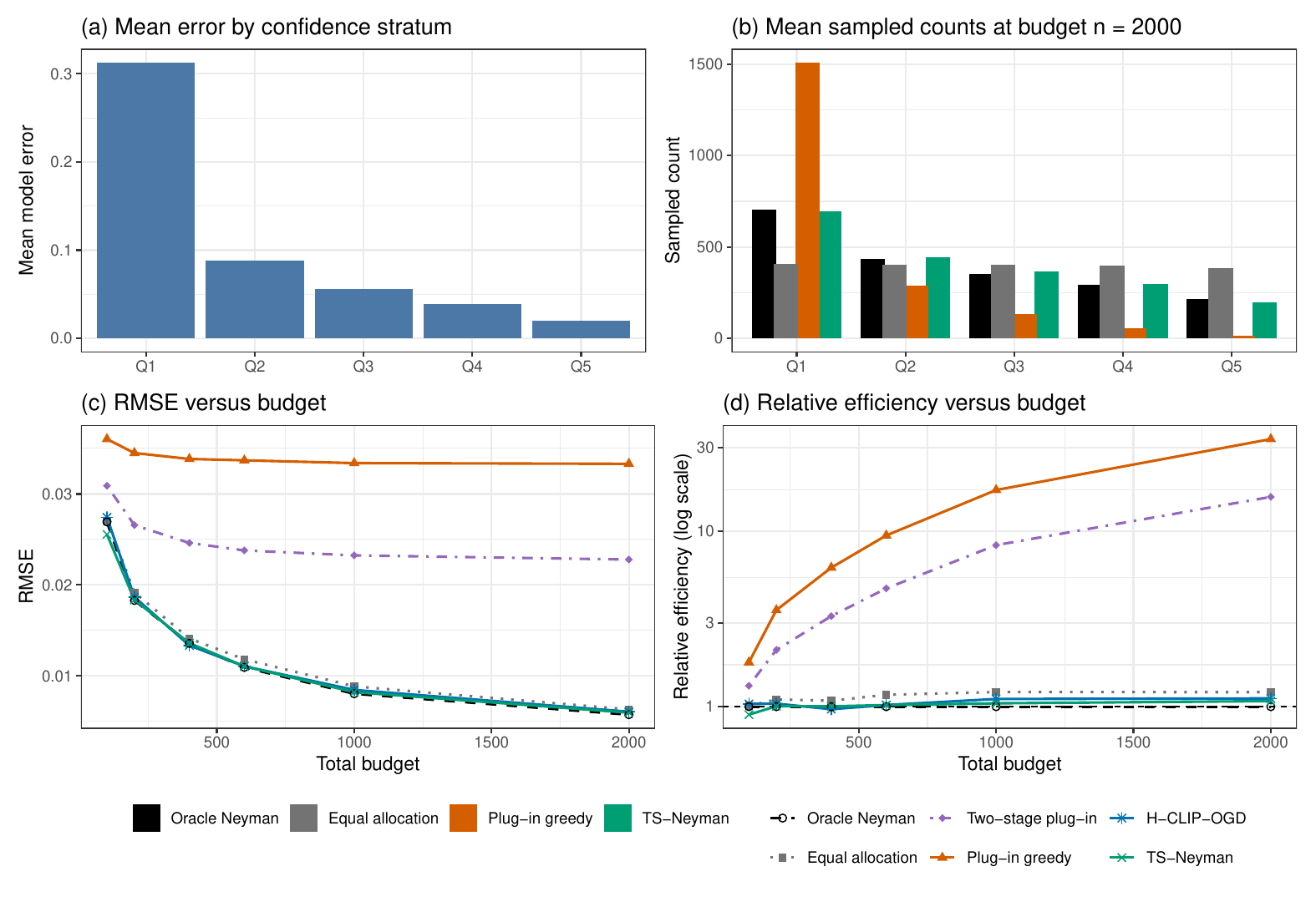}
\caption{CivilComments real-data replay with confidence-quintile strata and $Y=\mathbf{1}\{\text{model error}\}$. (a) Mean model error on the fixed test population by confidence stratum. (b) Mean sampled counts by stratum at budget $n=2000$. (c) RMSE versus budget. (d) Relative efficiency versus budget on a log scale.}
\label{fig:civilcomments_main}
\end{figure*}

\FloatBarrier

Appendix~\ref{app:additional-results} reports the common-pilot warm-start and prior-sensitivity diagnostics.

\section{Limitations and broader impact}
\label{sec:limitations}

Our theoretical analysis is asymptotic and assumes a fixed finite number of strata, positive within-stratum variances, iid sampling within each stratum in the asymptotic model, immediate observations, and one-by-one adaptive sampling.  The inverse-$\chi^2$ posterior is used as an algorithmic working model rather than as a substantive generative assumption.  The prior-sensitivity experiment suggests robustness to moderate choices, but very small pilots, heavy tails, delayed feedback, batched sampling, unequal costs, or severe working-model misspecification can still affect early allocations.  The empirical studies consist of two controlled five-stratum benchmarks together with one offline CivilComments replay rather than a broad evaluation across many ML datasets or deployment logs.  Computationally, TS--Neyman requires $O(H)$ posterior draws and index evaluations per round, so it scales linearly in the number of strata.

The potential positive impact is that adaptive Neyman allocation can reduce the number of labels, human ratings, expert reviews, or simulator calls needed to estimate average model performance, making careful model evaluation and subgroup or difficulty-bin auditing cheaper.  A potential risk is that optimizing a single average metric can allocate relatively little budget to low-variance or low-weight groups, even when those groups are substantively important for fairness, safety, or policy reporting.  In deployment-facing evaluations, practitioners should combine adaptive allocation with minimum per-stratum quotas, transparent reporting of allocation counts and uncertainty, and prespecified subgroup metrics that reflect the evaluation question rather than relying solely on aggregate precision.  The reported experiments use public benchmark data and offline replay rather than new private-data collection or human-subject studies.

\bibliographystyle{plainnat}
\bibliography{references}

@article{neyman1934two,
  author  = {Neyman, Jerzy},
  title   = {On the two different aspects of the representative method: the method of stratified sampling and the method of purposive selection},
  journal = {Journal of the Royal Statistical Society},
  year    = {1934},
  volume  = {97},
  number  = {4},
  pages   = {558--625},
  doi     = {10.2307/2342192}
}

@book{cochran1977sampling,
  author    = {Cochran, William G.},
  title     = {Sampling Techniques},
  edition   = {3rd},
  publisher = {Wiley},
  address   = {New York},
  year      = {1977}
}

@article{thompson1933likelihood,
  author  = {Thompson, William R.},
  title   = {On the likelihood that one unknown probability exceeds another in view of the evidence of two samples},
  journal = {Biometrika},
  year    = {1933},
  volume  = {25},
  number  = {3--4},
  pages   = {285--294},
  doi     = {10.2307/2332286}
}

@article{russo2018tutorial,
  author  = {Russo, Daniel J. and Van Roy, Benjamin and Kazerouni, Abbas and Osband, Ian and Wen, Zheng},
  title   = {A tutorial on {T}hompson sampling},
  journal = {Foundations and Trends in Machine Learning},
  year    = {2018},
  volume  = {11},
  number  = {1},
  pages   = {1--96},
  doi     = {10.1561/2200000070}
}

@article{EtoreJourdain2010Adaptive,
  author  = {{\'E}tor{\'e}, Pierre and Jourdain, Benjamin},
  title   = {Adaptive optimal allocation in stratified sampling methods},
  journal = {Methodology and Computing in Applied Probability},
  year    = {2010},
  volume  = {12},
  number  = {3},
  pages   = {335--360},
  doi     = {10.1007/s11009-008-9108-0}
}

@article{Zacks1970BayesianDesign,
  author  = {Zacks, S.},
  title   = {Bayesian design of single and double stratified sampling for estimating proportion in finite population},
  journal = {Technometrics},
  year    = {1970},
  volume  = {12},
  number  = {1},
  pages   = {119--130},
  doi     = {10.1080/00401706.1970.10488639}
}

@inproceedings{DaiGraduHarshaw2023ClipOGD,
  author    = {Dai, Jessica and Gradu, Paula and Harshaw, Christopher},
  title     = {{CLIP-OGD}: An experimental design for adaptive {N}eyman allocation in sequential experiments},
  booktitle = {Advances in Neural Information Processing Systems},
  volume    = {36},
  year      = {2023},
  pages     = {32235--32269},
  url       = {https://openreview.net/forum?id=2Xqvk2KVAq}
}

@article{YuZhaiSra2019NearOptimal,
  author  = {Yu, Tiancheng and Zhai, Xiyu and Sra, Suvrit},
  title   = {Near optimal stratified sampling},
  journal = {arXiv preprint arXiv:1906.11289},
  year    = {2019},
  doi     = {10.48550/arXiv.1906.11289},
  url     = {https://arxiv.org/abs/1906.11289}
}

@book{Durrett2019,
  author    = {Durrett, Rick},
  title     = {Probability: Theory and Examples},
  edition   = {5},
  publisher = {Cambridge University Press},
  year      = {2019}
}

@book{HallHeyde1980,
  author    = {Hall, Peter and Heyde, C. C.},
  title     = {Martingale Limit Theory and Its Application},
  publisher = {Academic Press},
  year      = {1980}
}

@article{Kadane2005OptimalDynamic,
  author  = {Kadane, Joseph B.},
  title   = {Optimal Dynamic Sample Allocation Among Strata},
  journal = {Journal of Official Statistics},
  year    = {2005},
  volume  = {21},
  number  = {4},
  pages   = {531--541}
}

@article{Wright2017ExactOptimal,
  author  = {Wright, Tommy},
  title   = {Exact optimal sample allocation: More efficient than Neyman},
  journal = {Statistics \& Probability Letters},
  year    = {2017},
  volume  = {129},
  pages   = {50--57},
  doi     = {10.1016/j.spl.2017.04.026}
}

@article{Wright2020GeneralExact,
  author  = {Wright, Tommy},
  title   = {A general exact optimal sample allocation algorithm: With bounded cost and bounded sample sizes},
  journal = {Statistics \& Probability Letters},
  year    = {2020},
  volume  = {165},
  pages   = {108829},
  doi     = {10.1016/j.spl.2020.108829}
}

@article{HahnHiranoKarlan2011,
  author  = {Hahn, Jinyong and Hirano, Keisuke and Karlan, Dean},
  title   = {Adaptive Experimental Design Using the Propensity Score},
  journal = {Journal of Business \& Economic Statistics},
  year    = {2011},
  volume  = {29},
  number  = {1},
  pages   = {96--108},
  doi     = {10.1198/jbes.2009.08161}
}

@misc{BlackwellPashleyValentino2023,
  author       = {Blackwell, Matthew and Pashley, Nicole E. and Valentino, Dominic},
  title        = {Batch Adaptive Designs to Improve Efficiency in Social Science Experiments},
  year         = {2023},
  howpublished = {Working paper},
  url          = {https://www.mattblackwell.org/files/papers/batch_adaptive.pdf}
}

@article{Zhao2023AdaptiveNeyman,
  author  = {Zhao, Jinglong},
  title   = {Adaptive Neyman Allocation},
  journal = {arXiv preprint arXiv:2309.08808},
  year    = {2023},
  url     = {https://arxiv.org/abs/2309.08808}
}

@article{horvitz1952generalization,
  author  = {Horvitz, D. G. and Thompson, D. J.},
  title   = {A generalization of sampling without replacement from a finite universe},
  journal = {Journal of the American Statistical Association},
  year    = {1952},
  volume  = {47},
  number  = {260},
  pages   = {663--685},
  doi     = {10.1080/01621459.1952.10483446}
}

@article{carpentier2015adaptive,
  author  = {Carpentier, Alexandra and Munos, R{\'e}mi and Antos, Andr{\'a}s},
  title   = {Adaptive Strategy for Stratified Monte Carlo Sampling},
  journal = {Journal of Machine Learning Research},
  year    = {2015},
  volume  = {16},
  number  = {68},
  pages   = {2231--2271},
  url     = {https://jmlr.org/papers/v16/carpentier15a.html}
}

@article{salehi2010efficient,
  author  = {Salehi, Mohammad and Moradi, Mohammad and Brown, Jennifer A. and Smith, David R.},
  title   = {Efficient estimators for adaptive stratified sequential sampling},
  journal = {Journal of Statistical Computation and Simulation},
  year    = {2010},
  volume  = {80},
  number  = {10},
  pages   = {1163--1179},
  doi     = {10.1080/00949650903005664}
}

@article{groves2006responsive,
  author  = {Groves, Robert M. and Heeringa, Steven G.},
  title   = {Responsive Design for Household Surveys: Tools for Actively Controlling Survey Errors and Costs},
  journal = {Journal of the Royal Statistical Society Series A: Statistics in Society},
  year    = {2006},
  volume  = {169},
  number  = {3},
  pages   = {439--457},
  doi     = {10.1111/j.1467-985X.2006.00423.x}
}

@article{schouten2013optimizing,
  author  = {Schouten, Barry and Calinescu, Melania and Luiten, Annemieke},
  title   = {Optimizing quality of response through adaptive survey designs},
  journal = {Survey Methodology},
  year    = {2013},
  volume  = {39},
  number  = {1},
  pages   = {29--58},
  url     = {https://www150.statcan.gc.ca/n1/pub/12-001-x/2013001/article/11824-eng.htm}
}

@article{yang2025optimall,
  author  = {Yang, Jasper B. and Lumley, Thomas and Shepherd, Bryan E. and Shaw, Pamela A.},
  title   = {Optimum Allocation for Adaptive Multi-Wave Sampling in {R}: The {R} Package {optimall}},
  journal = {Journal of Statistical Software},
  year    = {2025},
  volume  = {114},
  number  = {10},
  pages   = {1--31},
  doi     = {10.18637/jss.v114.i10},
  url     = {https://www.jstatsoft.org/article/view/v114i10}
}

@article{huntington1928apportionment,
  author  = {Huntington, Edward V.},
  title   = {The Apportionment of Representatives in Congress},
  journal = {Transactions of the American Mathematical Society},
  volume  = {30}, number = {1}, pages = {85--110}, year = {1928},
  doi     = {10.2307/1989268}
}

@book{balinski1982fair,
  author    = {Balinski, Michel L. and Young, H. Peyton},
  title     = {Fair Representation: Meeting the Ideal of One Man, One Vote},
  publisher = {Yale University Press}, year = {1982}, address = {New Haven}
}

@inproceedings{koh2021wilds,
  title={WILDS: A Benchmark of in-the-Wild Distribution Shifts},
  author={Koh, Pang Wei and Sagawa, Shiori and Marklund, Henrik and Xie, Sang Michael and Zhang, Marvin and Balsubramani, Akshay and Hu, Weihua and Yasunaga, Michihiro and Philipose, Matt and Beery, Sara and others},
  booktitle={Proceedings of the 38th International Conference on Machine Learning},
  year={2021}
}

@inproceedings{borkan2019nuanced,
  title={Nuanced Metrics for Measuring Unintended Bias with Real Data for Text Classification},
  author={Borkan, Daniel and Sorensen, Jeffrey and Thain, Nithum and Dixon, Lucas},
  booktitle={Companion Proceedings of the 2019 World Wide Web Conference},
  year={2019}
}

\appendix

\section{Additional numerical results}
\label{app:additional-results}

\subsection{CivilComments replay details}
\label{app:civilcomments-appendix}

Table~\ref{tab:civilcomments_re_appendix} reports the full per-budget relative-efficiency summary for the CivilComments confidence-quintile replay. Oracle Neyman is normalized to 1 and omitted.

\begin{table}[t]
\centering
\small
\caption{Relative efficiency in the CivilComments replay experiment. Oracle Neyman is normalized to 1 and omitted.}
\label{tab:civilcomments_re_appendix}
\begin{tabular}{rrrrrr}
\toprule
Budget & Equal allocation & Two-stage plug-in & Plug-in greedy & H-CLIP-OGD & TS--Neyman\\
\midrule
100 & 1.011 & 1.316 & 1.790 & 1.037 & 0.900\\
200 & 1.099 & 2.109 & 3.555 & 1.040 & 1.006\\
400 & 1.081 & 3.285 & 6.211 & 0.967 & 1.004\\
600 & 1.169 & 4.727 & 9.478 & 1.022 & 1.021\\
1000 & 1.212 & 8.347 & 17.215 & 1.106 & 1.045\\
2000 & 1.210 & 15.731 & 33.604 & 1.114 & 1.076\\
\bottomrule
\end{tabular}
\end{table}

\subsection{Warm-start robustness details}
\label{app:robustness-details}

Table~\ref{tab:robustness_common_pilot_full} reports the full common-pilot warm-start summary with Monte Carlo bias, SD, RMSE, and RE.  The common-pilot design uses the same $m_0\in\{2,5,10\}$ for plug-in greedy, H-CLIP-OGD, and TS--Neyman.

\begin{table}[t]
\centering
\small
\caption{Full common-pilot warm-start summary. RE is relative to Oracle Neyman.}
\label{tab:robustness_common_pilot_full}
\begin{tabular}{llrrrr}
\toprule
Setting & Method & Bias & SD & RMSE & RE\\
\midrule
Gaussian ($m_0=2$) & Oracle Neyman & 0.001 & 0.173 & 0.173 & 1.000\\
Gaussian ($m_0=2$) & Plug-in greedy & -0.001 & 0.224 & 0.224 & 1.679\\
Gaussian ($m_0=2$) & H-CLIP-OGD & 0.002 & 0.180 & 0.180 & 1.088\\
Gaussian ($m_0=2$) & TS-Neyman & 0.002 & 0.173 & 0.173 & 1.001\\
Gaussian ($m_0=5$) & Oracle Neyman & -0.007 & 0.175 & 0.175 & 1.000\\
Gaussian ($m_0=5$) & Plug-in greedy & -0.007 & 0.175 & 0.175 & 1.000\\
Gaussian ($m_0=5$) & H-CLIP-OGD & -0.006 & 0.180 & 0.180 & 1.065\\
Gaussian ($m_0=5$) & TS-Neyman & -0.006 & 0.174 & 0.175 & 0.998\\
Gaussian ($m_0=10$) & Oracle Neyman & -0.001 & 0.171 & 0.171 & 1.000\\
Gaussian ($m_0=10$) & Plug-in greedy & -0.001 & 0.171 & 0.171 & 1.001\\
Gaussian ($m_0=10$) & H-CLIP-OGD & 0.000 & 0.175 & 0.175 & 1.048\\
Gaussian ($m_0=10$) & TS-Neyman & -0.001 & 0.170 & 0.170 & 0.993\\
Uniform ($m_0=2$) & Oracle Neyman & 0.000 & 0.173 & 0.173 & 1.000\\
Uniform ($m_0=2$) & Plug-in greedy & -0.002 & 0.299 & 0.298 & 2.962\\
Uniform ($m_0=2$) & H-CLIP-OGD & -0.001 & 0.191 & 0.191 & 1.211\\
Uniform ($m_0=2$) & TS-Neyman & -0.001 & 0.175 & 0.175 & 1.016\\
Uniform ($m_0=5$) & Oracle Neyman & -0.004 & 0.176 & 0.176 & 1.000\\
Uniform ($m_0=5$) & Plug-in greedy & -0.005 & 0.178 & 0.178 & 1.018\\
Uniform ($m_0=5$) & H-CLIP-OGD & -0.002 & 0.181 & 0.181 & 1.055\\
Uniform ($m_0=5$) & TS-Neyman & -0.005 & 0.177 & 0.177 & 1.004\\
Uniform ($m_0=10$) & Oracle Neyman & 0.003 & 0.176 & 0.176 & 1.000\\
Uniform ($m_0=10$) & Plug-in greedy & 0.003 & 0.177 & 0.177 & 1.004\\
Uniform ($m_0=10$) & H-CLIP-OGD & 0.002 & 0.179 & 0.178 & 1.024\\
Uniform ($m_0=10$) & TS-Neyman & 0.003 & 0.176 & 0.176 & 0.999\\
\bottomrule
\end{tabular}
\end{table}

\subsection{Prior-sensitivity details}
\label{app:prior-sensitivity}

Table~\ref{tab:prior_sensitivity} reports the TS--Neyman prior-sensitivity grid over $(\nu_0,s_0^2)$ in the sparse-pilot regime.  The row with $s_0^2=\widehat s^2_{\mathrm{pilot}}$ uses the pooled pilot variance and is intended as a scale-adaptive default.

\begin{table}[t]
\centering
\caption{Prior sensitivity of TS-Neyman in the sparse-pilot setting. RE is relative to Oracle Neyman.}
\label{tab:prior_sensitivity}
\begin{tabular}{llrrrr}
\toprule
$\nu_0$ & $s_0^2$ & Bias & SD & RMSE & RE\\
\midrule
0.5 & 0.25 & 0.007 & 0.170 & 0.170 & 0.943\\
2 & 0.25 & 0.006 & 0.179 & 0.179 & 1.040\\
5 & 0.25 & 0.004 & 0.176 & 0.176 & 1.011\\
0.5 & 1 & 0.000 & 0.175 & 0.174 & 0.989\\
2 & 1 & 0.004 & 0.176 & 0.176 & 1.002\\
5 & 1 & 0.003 & 0.175 & 0.175 & 0.996\\
0.5 & 4 & 0.006 & 0.175 & 0.175 & 0.997\\
2 & 4 & 0.001 & 0.171 & 0.171 & 0.947\\
5 & 4 & 0.004 & 0.176 & 0.176 & 1.012\\
0.5 & $\widehat s^2_{\mathrm{pilot}}$ & 0.002 & 0.179 & 0.179 & 1.040\\
2 & $\widehat s^2_{\mathrm{pilot}}$ & 0.004 & 0.176 & 0.176 & 1.007\\
5 & $\widehat s^2_{\mathrm{pilot}}$ & 0.009 & 0.177 & 0.177 & 1.018\\
\bottomrule
\end{tabular}
\end{table}

\subsection{Reproducibility details}
\label{app:repro}

The experiment grid is implemented in base R and uses no non-base packages. The reported run used 2000 Monte Carlo repetitions, 500 bootstrap repetitions for the budget-scaling confidence intervals, and PSOCK parallelism with the default worker count set to the detected core count minus one. The recorded wall-clock times were 200.8 seconds for the bounded-loss budget-scaling study, 197.4 seconds for the binary model-error benchmark, 127.6 seconds for the common-pilot warm-start study, and 124.8 seconds for the prior-sensitivity study, for a total of 650.6 seconds. The run was performed in R 4.5.2 on macOS Tahoe 26.4.1; the output directory stores the exact tables, figures, raw \texttt{.rds} summaries, \texttt{timing.csv}, and \texttt{sessionInfo.txt}.

\section{H-CLIP-OGD construction}
\label{app:hclipogd}

Section~\ref{sec:hclipogd} introduces $H$-CLIP-OGD as an empirical comparator inspired by CLIP-OGD~\citep{DaiGraduHarshaw2023ClipOGD}. Here we give the full construction.

\paragraph{Naming caveat.}
We call the algorithm $H$-CLIP-OGD for brevity. The name should be read as ``CLIP-OGD-inspired projected-gradient baseline for $H$ strata'' rather than as a claim that the algorithm is the same method as CLIP-OGD with the same guarantees. CLIP-OGD~\citep{DaiGraduHarshaw2023ClipOGD} is a two-arm sequential-experiment method whose decision variable is a single treatment probability and whose gradient estimator is inverse-probability-weighted; it comes with a $\widetilde{\mathcal{O}}(\sqrt{T})$ Neyman-regret guarantee in that setting. Our $H$-stratum baseline replaces the two-arm probability by a simplex-valued design variable on $\Delta_H$ and uses a plug-in gradient. We do not claim a corresponding finite-sample regret theorem; $H$-CLIP-OGD is included only as a fair empirical comparator, and all formal results in Section~\ref{sec:theory} concern \textsc{TS--Neyman} alone.

\paragraph{Continuous relaxation.}
For the gradient algorithm we replace the exact integer-allocation objective $V(n; S)$ in \eqref{eq:proxy} by the continuous proxy
\begin{equation}
F(p; s^2) \;=\; \sum_{h=1}^H \frac{(N_h/N)^2 \, s_h^2}{p_h}, \qquad p \in \Delta_H,
\label{eq:fproxy}
\end{equation}
where $s_h^2$ is the current sample variance in stratum $h$ and $\Delta_H$ is the simplex. With $p_h = n_h / T$, this differs from $V(n; S)$ only by the constant factor $N^{-2}$, so the two objectives have the same minimizer in proportions: the Neyman target $p_h^\star = N_h \sigma_h / \sum_g N_g \sigma_g$.

The marginal-gain index used by \textsc{TS--Neyman} and the continuous gradient of $F$ used by $H$-CLIP-OGD encode the same first-order information about $V$, in discrete and continuous geometries respectively. Writing $\phi_h(x) = N_h^2 \sigma_h^2 / x$, the marginal gain in \eqref{eq:marginal-gain} is the negative forward difference
\[
\Delta_h(n_h; \sigma_h^2) \;=\; \phi_h(n_h) - \phi_h(n_h + 1) \;=\; \frac{N_h^2 \sigma_h^2}{n_h (n_h + 1)},
\]
while the continuous derivative is $\phi_h'(x) = -N_h^2 \sigma_h^2 / x^2$, agreeing with the forward difference up to $\mathcal{O}(n_h^{-3})$ for large $n_h$. Substituting $n_h = T p_h$ gives $\Delta_h(T p_h; \sigma_h^2) \approx -T^{-2} \, \partial F / \partial p_h$. \textsc{TS--Neyman} takes integer steps at the posterior-randomized argmax of these increments, while $H$-CLIP-OGD takes small continuous steps in $p_t$ and then samples from the updated distribution; both climb the same Neyman objective, with different step shapes.

\paragraph{Plug-in gradient.}
The unknown variances $\sigma_h^2$ are replaced by the running sample variances $s_h^2(t)$, giving the plug-in gradient
\begin{equation}
\widehat{g}_t^{(p)}(h) \;=\; -\,\frac{(N_h/N)^2 \, s_h^2(t)}{p_t(h)^2}.
\label{eq:fgrad}
\end{equation}
This is consistent (but not unbiased) on $\{ n_h(t) \to \infty \}$, which the clipping schedule below guarantees.

\paragraph{Clipping schedule.}
The clipping level
\begin{equation}
\varepsilon_t \;=\; \min\!\left( \frac{1}{2H},\; c_\varepsilon \, t^{-\beta} \right), \qquad \beta \in (0, 1],
\label{eq:epsilon}
\end{equation}
is chosen to satisfy three properties: (i) $H \varepsilon_t \leq 1/2$ at every $t$, so the mixture below is well-defined; (ii) $\varepsilon_t \to 0$, so the constraint relaxes asymptotically; and (iii) $\sum_t \varepsilon_t = \infty$ when $\beta \leq 1$, which is what a persistent-exploration argument needs. The role of $\varepsilon_t$ is to enforce a uniform lower bound $p_t(h) \geq \varepsilon_t$ on every stratum at every round, so that no stratum's variance estimate can freeze and the gradient $\widehat{g}_t^{(p)}$ stays bounded.

\paragraph{$\varepsilon$-mixture parametrization.}
To enforce the lower bound while keeping the projection step standard, we parametrize $p_t$ as a mixture of an unconstrained simplex iterate $q_t \in \Delta_H$ with the uniform distribution:
\begin{equation}
p_t \;=\; (1 - H \varepsilon_t)\, q_t + \varepsilon_t \, \mathbf{1},
\label{eq:epsmixture}
\end{equation}
where $\mathbf{1}$ is the all-ones vector. This guarantees $p_t(h) \geq \varepsilon_t$ exactly, with no further projection beyond standard simplex projection on $q_t$. The chain-rule factor $\partial p_t / \partial q_t = (1 - H \varepsilon_t)\, I$ is carried explicitly into the gradient step:
\begin{equation}
q_{t+1} \;=\; \Pi_{\Delta_H}\!\left( q_t - \eta_t (1 - H \varepsilon_t)\, \widehat{g}_t^{(p)} \right),
\label{eq:fpgd}
\end{equation}
with $\Pi_{\Delta_H}$ the Euclidean projection onto the standard simplex. The factor $(1 - H \varepsilon_t)$ tends to one as $\varepsilon_t \to 0$ and could be absorbed into the step size, but we keep it explicit for correctness.

\paragraph{Randomized action and pseudocode.}
At each round, $H$-CLIP-OGD draws $A_t \sim p_t$ and observes the next within-stratum value, updating only that stratum's sample variance. The action is drawn at random from $p_t$ rather than by $\arg\max$; this is the same design choice that motivates posterior sampling over plug-in greedy in Section~\ref{sec:tsneyman}, namely to avoid the lock-in failure mode in which a stratum is permanently selected or excluded on the basis of a noisy early variance estimate. The full procedure is summarized in Algorithm~\ref{alg:hclipogd}.

\begin{algorithm}[h]
\caption{$H$-CLIP-OGD}
\label{alg:hclipogd}
\begin{algorithmic}[1]
\Require Stratum sizes $\{N_h\}_{h=1}^H$, total budget $n$, warm-start $m_0 \geq 2$, hyperparameters $(\{\eta_t\}, \beta, c_\varepsilon)$.
\For{$h = 1, \ldots, H$}
\State Collect $m_0$ observations from stratum $h$ and initialize $n_h \leftarrow m_0$.
\EndFor
\State Initialize $q_{Hm_0 + 1} \leftarrow \mathbf{1} / H$.
\For{$t = H m_0 + 1, \ldots, n$}
\State Set $\varepsilon_t = \min\{1/(2H),\, c_\varepsilon\, t^{-\beta}\}$ and form $p_t = (1 - H \varepsilon_t) q_t + \varepsilon_t \mathbf{1}$.
\State For each $h$, set $\widehat{g}_t^{(p)}(h) = -(N_h/N)^2 s_h^2(t) / p_t(h)^2$.
\State Draw $A_t \sim p_t$.
\State Sample one additional unit from stratum $A_t$, update $s_{A_t}^2$, and increment $n_{A_t}$.
\State $q_{t+1} \leftarrow \Pi_{\Delta_H}\!\left( q_t - \eta_t (1 - H \varepsilon_t)\, \widehat{g}_t^{(p)} \right)$.
\EndFor 
\State 
\Return Final allocation $(n_1, \ldots, n_H)$ and the stratified sample.
\end{algorithmic}
\end{algorithm}

\paragraph{Hyperparameters.}
The experiments use values of $c_\varepsilon$, $\beta$, and a step-size schedule $\eta_t$ recorded in the reproducibility code. A workable default for fixed-budget runs is $\eta_t = \eta = c_\eta\, n^{-1/2}$, constant in $t$ but scaled to the horizon $n$, since the gradient bound $\|\widehat{g}_t^{(p)}\|_\infty \lesssim \max_h N_h^2 s_h^2(t) / \varepsilon_t^2$ depends on the clipping level.

\paragraph{Specialization to $H = 2$.}
For $H = 2$, parametrize $p = (p, 1 - p)$. The mixture $p_t$ lies in the clipped interval $[\varepsilon_t, 1 - \varepsilon_t]$, the objective reduces to
\[
F(p; s^2) \;=\; \frac{(N_1/N)^2 s_1^2}{p} + \frac{(N_2/N)^2 s_2^2}{1 - p},
\]
and Algorithm~\ref{alg:hclipogd} becomes clipped projected stochastic gradient descent on this scalar function with a plug-in gradient. This is structurally similar to CLIP-OGD~\citep{DaiGraduHarshaw2023ClipOGD}, with two principal differences: the action space here is two strata in a sequential survey rather than treatment and control in a sequential experiment, and the gradient estimator is plug-in rather than the inverse-probability-weighted estimator that gives CLIP-OGD its unbiasedness on each round.

\section{Proofs}
\label{app:proofs}

Throughout this appendix, $H\ge2$ is fixed and finite.
Write $[H]=\{1,\dots,H\}$, $w_h=N_h\sigma_h$, $W=\sum_{g=1}^H w_g$, and $p_h=w_h/W$.
After initialization, each stratum has received $m_0\ge2$ observations and the initial total sample size is $t_0:=Hm_0$.
For $t\ge t_0$, let $n_h(t)$ be the number of observations collected from stratum $h$ by time $t$, so that $\sum_h n_h(t)=t$.
Define
\[
x_h(t):=\frac{n_h(t)}{t},
\qquad
r_h(t):=\frac{x_h(t)}{p_h}.
\]
The vector $r(t)$ has $p$-weighted mean one: $\sum_h p_h r_h(t)=\sum_h x_h(t)=1$.
Let $(\mathcal F_t)_{t\ge t_0}$ be the natural filtration generated by the initial observations and all subsequent allocation decisions and sampled responses.
The predictable allocation probabilities are
\[
\pi_{t+1,h}:=\Prob(A_{t+1}=h\mid \mathcal F_t),
\qquad h\in[H],
\]
and the compensators are
\[
\Lambda_h(t):=m_0+\sum_{s=t_0}^{t-1}\pi_{s+1,h}.
\]
Define also
\[
R_h(t):=\sum_{j=1}^{n_h(t)}Y_{hj},
\qquad
\bar Y_h(t):=\frac{R_h(t)}{n_h(t)}.
\]

\subsection{Basic martingale decompositions}

\begin{proposition}[Basic martingale decompositions]
\label{prop:martingale}
For each $h\in[H]$,
\[
M_h^N(t):=n_h(t)-\Lambda_h(t),
\qquad
M_h^Y(t):=R_h(t)-\mu_h n_h(t)
\]
are square-integrable $(\mathcal F_t)$-martingales. Moreover,
\[
\langle M_h^N\rangle(t)=\sum_{s=t_0}^{t-1}\pi_{s+1,h}(1-\pi_{s+1,h}),
\qquad
\langle M_h^Y\rangle(t)=\sigma_h^2\sum_{s=t_0}^{t-1}\pi_{s+1,h}.
\]
\end{proposition}

\begin{proof}
For $M_h^N$,
\[
M_h^N(t+1)-M_h^N(t)=\1\{A_{t+1}=h\}-\pi_{t+1,h},
\]
whose conditional expectation given $\mathcal F_t$ is zero.
For $M_h^Y$,
\[
M_h^Y(t+1)-M_h^Y(t)
=
\bigl(Y_{h,n_h(t)+1}-\mu_h\bigr)\1\{A_{t+1}=h\}.
\]
Conditional on $\mathcal F_t$, the next unused within-stratum draw has mean $\mu_h$ and is independent of the TS randomization, so the conditional expectation is zero.
The predictable quadratic-variation formulas follow by taking conditional second moments of the increments.
\end{proof}

\subsection{Posterior concentration along infinitely sampled strata}

\begin{lemma}[Working-posterior concentration]
\label{lem:posterior-concentration}
Fix $h\in[H]$.
On any event on which $n_h(t)\to\infty$, one has
\[
s_h^2(t)\toas \sigma_h^2,
\qquad
\tau_h^2(t)\toas \sigma_h^2,
\qquad
\nu_h(t)\to\infty,
\]
where
\[
\nu_h(t)=\nu_0+n_h(t)-1,
\qquad
\tau_h^2(t)=\frac{\nu_0s_0^2+(n_h(t)-1)s_h^2(t)}{\nu_0+n_h(t)-1}.
\]
Moreover, for every $\delta>0$,
\[
\Prob\bigl(|\wt\sigma_h^2(t)-\sigma_h^2|>\delta\mid\mathcal F_t\bigr)\toas0.
\]
\end{lemma}

\begin{proof}
On $\{n_h(t)\to\infty\}$, strong consistency of the sample variance gives $s_h^2(t)\toas\sigma_h^2$ because $\E[(Y_{h1}-\mu_h)^2]=\sigma_h^2<\infty$.
The displayed formula for $\tau_h^2(t)$ then yields $\tau_h^2(t)\toas\sigma_h^2$, while $\nu_h(t)\to\infty$.
Conditional on $\mathcal F_t$,
\[
\wt\sigma_h^2(t)=\tau_h^2(t)\frac{\nu_h(t)}{Z_t},
\qquad
Z_t\mid\mathcal F_t\sim\chi^2_{\nu_h(t)}.
\]
Fix $\delta>0$ and choose $\eta\in(0,1/2)$ so small that $2\eta(\sigma_h^2+1)+\eta<\delta$.
On the event
\[
\Bigl\{|\tau_h^2(t)-\sigma_h^2|\le\eta\Bigr\}
\cap
\Bigl\{\bigl|Z_t/\nu_h(t)-1\bigr|\le\eta\Bigr\},
\]
one has $\nu_h(t)/Z_t\le(1-\eta)^{-1}\le2$, and hence
\[
\bigl|\wt\sigma_h^2(t)-\sigma_h^2\bigr|
\le
\left|\tau_h^2(t)-\sigma_h^2\right|\frac{\nu_h(t)}{Z_t}
+\sigma_h^2\left|\frac{\nu_h(t)}{Z_t}-1\right|
<\delta.
\]
Thus
\[
\Prob\bigl(|\wt\sigma_h^2(t)-\sigma_h^2|>\delta\mid\mathcal F_t\bigr)
\le
\1\{|\tau_h^2(t)-\sigma_h^2|>\eta\}
+
\Prob\left(\left|\frac{\chi^2_{\nu_h(t)}}{\nu_h(t)}-1\right|>\eta\right),
\]
and the right-hand side converges almost surely to zero.
\end{proof}

\subsection{Persistent exploration}

\begin{proposition}[Persistent exploration]
\label{prop:infinite-sampling}
Under Assumption~\ref{ass:fixed-h},
\[
n_h(t)\toas\infty,
\qquad h=1,\dots,H.
\]
\end{proposition}

\begin{proof}
Fix $h\in[H]$ and suppose, toward a contradiction, that $n_h(t)$ is bounded with positive probability.
It is enough to work on an event $E_S$ of positive probability on which the set $S\subset[H]$ of strata sampled only finitely often is fixed and contains $h$.
Because counts are nondecreasing integers, on $E_S$ there is an almost surely finite random time $T_S$ after which $n_g(t)$ is constant for every $g\in S$, while $n_g(t)\to\infty$ for every $g\notin S$.

For $g\in S$, the working-posterior law of $\wt\Delta_g(t)$ is frozen after $T_S$ and has support $(0,\infty)$.
The posterior draws are conditionally independent across strata, so there exist an $\mathcal F_{T_S}$-measurable number $b>0$ and an $\mathcal F_{T_S}$-measurable number $\eta>0$ such that, for all $t\ge T_S$,
\[
\Prob\left(
\wt\Delta_h(t)>b\quad\text{and}\quad \wt\Delta_g(t)<b\ \forall g\in S\setminus\{h\}
\,\middle|\,\mathcal F_t
\right)\ge\eta
\qquad\text{on }E_S.
\]
For $g\notin S$, Lemma~\ref{lem:posterior-concentration} and $n_g(t)\to\infty$ imply $\wt\Delta_g(t)\to0$ in conditional probability, because the denominator is of order $n_g(t)^2$ while the posterior variance draw is tight.
Hence, on $E_S$, for all sufficiently large $t$,
\[
\Prob\left(\max_{g\notin S}\wt\Delta_g(t)>b\,\middle|\,\mathcal F_t\right)\le\eta/2.
\]
Consequently, for all sufficiently large $t$ on $E_S$,
\[
\Prob(A_{t+1}=h\mid\mathcal F_t)\ge\eta/2.
\]
The conditional Borel--Cantelli lemma \citep[Theorem~5.3.2]{Durrett2019} then implies that $A_{t+1}=h$ occurs infinitely often on $E_S$, contradicting $h\in S$.
Thus no stratum can be sampled only finitely often.
\end{proof}

\subsection{A fixed-\texorpdfstring{$H$}{H} ranking lemma}

\begin{lemma}[Selection of under-allocated strata]
\label{lem:underallocated-selection}
Let
\[
Q_t:=\sum_{h=1}^H p_h\{r_h(t)-1\}^2.
\]
For every $\varepsilon>0$, there exist constants $a_\varepsilon>0$, $\gamma_\varepsilon>0$, a deterministic integer $T_\varepsilon$, and an $\mathcal F_t$-measurable sequence $\rho_t^{(\varepsilon)}\ge0$ with $\rho_t^{(\varepsilon)}\toas0$ such that for all $t\ge T_\varepsilon$, on the event
\[
\{Q_t\ge\varepsilon\}\cap\{\rho_t^{(\varepsilon)}\le\gamma_\varepsilon\},
\]
one has
\[
\sum_{h=1}^H \pi_{t+1,h}r_h(t)\le1-a_\varepsilon.
\]
\end{lemma}

\begin{proof}
Let $B:=\max_h 1/p_h$, so $0\le r_h(t)\le B$ for all $h,t$.
By compactness of the set $\{r\in[0,B]^H:\sum_h p_hr_h=1,\sum_h p_h(r_h-1)^2\ge\varepsilon\}$, there exists $c_\varepsilon>0$ such that every $r$ in this set satisfies
\[
\min_h r_h\le1-c_\varepsilon.
\]
Define the under-allocated set
\[
U_t:=\{h:r_h(t)\le1-c_\varepsilon/2\}.
\]
On $\{Q_t\ge\varepsilon\}$, this set is nonempty.

Consider the true, nonrandom marginal-gain indices
\[
D_h(t):=\frac{w_h^2}{n_h(t)\{n_h(t)+1\}}.
\]
Since $w_h=Wp_h$ and $n_h(t)=tp_hr_h(t)$,
\[
D_h(t)=\frac{W^2}{t^2}\frac{1}{r_h(t)^2+r_h(t)/(tp_h)}.
\]
Thus, for all sufficiently large $t$, any stratum with $r_h(t)\le1-c_\varepsilon$ has a strictly larger true index than any stratum with $r_g(t)>1-c_\varepsilon/2$, uniformly over all feasible allocation vectors.
Therefore there are constants $\delta_\varepsilon>0$ and a deterministic integer $T_\varepsilon$ such that, for all $t\ge T_\varepsilon$, if all posterior variance draws satisfy
\[
|\wt\sigma_h^2(t)-\sigma_h^2|\le\delta_\varepsilon,
\qquad h=1,\dots,H,
\]
then the TS--Neyman maximizer belongs to $U_t$ whenever $Q_t\ge\varepsilon$.

Set
\[
\rho_t^{(\varepsilon)}
:=
\sum_{h=1}^H
\Prob\bigl(|\wt\sigma_h^2(t)-\sigma_h^2|>\delta_\varepsilon\mid\mathcal F_t\bigr).
\]
By Proposition~\ref{prop:infinite-sampling} and Lemma~\ref{lem:posterior-concentration}, $\rho_t^{(\varepsilon)}\toas0$.
On $\{Q_t\ge\varepsilon\}$, the preceding paragraph and a union bound imply
\[
\Prob(A_{t+1}\in U_t\mid\mathcal F_t)\ge1-\rho_t^{(\varepsilon)}.
\]
Since $r_h(t)\le1-c_\varepsilon/2$ for $h\in U_t$ and $r_h(t)\le B$ always,
\[
\sum_{h=1}^H\pi_{t+1,h}r_h(t)
\le
1-c_\varepsilon/2+B\rho_t^{(\varepsilon)}.
\]
The claim follows by taking $\gamma_\varepsilon=c_\varepsilon/(4B)$ and $a_\varepsilon=c_\varepsilon/4$.
\end{proof}

\subsection{Robbins--Siegmund lemma}

\begin{lemma}[Robbins--Siegmund almost-supermartingale theorem]
\label{lem:rs}
Let $(Z_t)$, $(\beta_t)$, $(\xi_t)$, and $(\zeta_t)$ be nonnegative, $\mathcal F_t$-adapted processes such that, for every $t\ge t_0$,
\[
\E[Z_{t+1}\mid\mathcal F_t]
\le
(1+\beta_t)Z_t+\xi_t-\zeta_t.
\]
If $\sum_t\beta_t<\infty$ and $\sum_t\xi_t<\infty$ almost surely, then $Z_t$ converges almost surely to a finite random variable and $\sum_t\zeta_t<\infty$ almost surely.
\end{lemma}

\subsection{Proof of Theorem~\ref{thm:consistency}}

\begin{proof}
Let
\[
Q_t=\sum_{h=1}^H p_h\{r_h(t)-1\}^2,
\qquad
r_h(t)=\frac{x_h(t)}{p_h}.
\]
Because $x_{t+1}=x_t+(e_{A_{t+1}}-x_t)/(t+1)$, where $e_j$ is the $j$th coordinate vector,
\[
r_h(t+1)
=
r_h(t)+\frac{\1\{A_{t+1}=h\}/p_h-r_h(t)}{t+1}.
\]
A direct expansion gives
\begin{align*}
\E[Q_{t+1}\mid\mathcal F_t]
&=
Q_t+
\frac{2}{t+1}
\left\{
\sum_{h=1}^H\pi_{t+1,h}r_h(t)-1-Q_t
\right\}
+O\!\left(\frac{1}{(t+1)^2}\right),
\end{align*}
where the $O((t+1)^{-2})$ term is deterministic and finite because $H$ is fixed and $p_{\min}:=\min_h p_h>0$.

Fix $\varepsilon>0$ and apply Lemma~\ref{lem:underallocated-selection}.
Let $G_t^{(\varepsilon)}=\{t\ge T_\varepsilon\}\cap\{\rho_t^{(\varepsilon)}\le\gamma_\varepsilon\}$.
Since $\rho_t^{(\varepsilon)}\toas0$, the complement of $G_t^{(\varepsilon)}$ occurs only finitely often almost surely.
Using the trivial bound $\sum_h\pi_{t+1,h}r_h(t)\le B$, the preceding display implies
\[
\E[Q_{t+1}\mid\mathcal F_t]
\le
Q_t+
\xi_t^{(\varepsilon)}-
\zeta_t^{(\varepsilon)},
\]
where $\sum_t\xi_t^{(\varepsilon)}<\infty$ almost surely and
\[
\zeta_t^{(\varepsilon)}
=
\frac{2a_\varepsilon}{t+1}
\1\{Q_t\ge\varepsilon\}\1_{G_t^{(\varepsilon)}}.
\]
By Lemma~\ref{lem:rs}, $Q_t$ converges almost surely and $\sum_t\zeta_t^{(\varepsilon)}<\infty$.
Since $G_t^{(\varepsilon)}$ holds eventually almost surely,
\[
\sum_{t=t_0}^\infty \frac{\1\{Q_t\ge\varepsilon\}}{t+1}<\infty
\qquad\text{a.s.}
\]
If the almost-sure limit of $Q_t$ were at least $\varepsilon$ on an event of positive probability, the last series would diverge on that event.
Therefore $\limsup_t Q_t<\varepsilon$ almost surely.
Taking a countable sequence $\varepsilon\downarrow0$ yields $Q_t\to0$ almost surely.
Since $p_h>0$ for every $h$, this is equivalent to $x_h(t)=n_h(t)/t\to p_h$ for all $h$.
\end{proof}

\subsection{Proof of Theorem~\ref{thm:efficiency}}

\begin{proof}
By Theorem~\ref{thm:consistency}, $n_h(t)=tp_h+o(t)$ almost surely.
Since $p_h=w_h/W$,
\[
\frac{w_h^2}{n_h(t)}
=
\frac{Ww_h}{t}+o\!\left(\frac{1}{t}\right)
\qquad\text{a.s.}
\]
Summing over $h$ gives
\[
V(\mathbf n(t))
=
\frac{W^2}{t}+o\!\left(\frac{1}{t}\right)
\qquad\text{a.s.}
\]
For any feasible integer allocation $\mathbf m$ with $\sum_hm_h=t$, Cauchy--Schwarz gives
\[
W^2=\left(\sum_h w_h\right)^2
\le
\left(\sum_h m_h\right)
\left(\sum_h\frac{w_h^2}{m_h}\right)
=tV(\mathbf m),
\]
so $V^\star(t)\ge W^2/t$.
For the matching upper bound, choose integers $\tilde m_h(t)$ with $\sum_h\tilde m_h(t)=t$ and $|\tilde m_h(t)-tp_h|\le C_H$ for a constant depending only on $H$.
Then
\[
\sum_h\frac{w_h^2}{\tilde m_h(t)}
=
\frac{W^2}{t}+O\!\left(\frac{1}{t^2}\right),
\]
and hence $V^\star(t)=W^2/t+O(t^{-2})$.
The ratio therefore converges to one almost surely.
\end{proof}

\subsection{Proof of Corollary~\ref{cor:clt}}

For clarity, we first state the estimator argument in a conditional form.

\begin{proposition}[Estimator consistency and asymptotic normality for the total]
\label{prop:estimator-total}
Fix any finite $H\ge2$.
Suppose each stratum produces iid observations with mean $\mu_h$ and variance $\sigma_h^2\in(0,\infty)$, and suppose the adaptive allocation satisfies
\[
\frac{n_h(t)}{t}\toas p_h\in(0,1),
\qquad h=1,\dots,H.
\]
Then, with $T=\sum_hN_h\mu_h$ and $\wh T_H(t)=\sum_hN_h\bar Y_h(t)$,
\[
\wh T_H(t)\toas T,
\]
and
\[
\sqrt t\{\wh T_H(t)-T\}
\rightsquig
\mathcal N\!\left(0,\sum_{h=1}^H\frac{N_h^2\sigma_h^2}{p_h}\right).
\]
\end{proposition}

\begin{proof}
For each $h$, write
\[
M_h^Y(t)=R_h(t)-\mu_hn_h(t).
\]
The increments of $M_h^Y$ have conditional second moments bounded by $\sigma_h^2$, so
\[
\sum_{t=t_0}^\infty
\frac{\E[(M_h^Y(t+1)-M_h^Y(t))^2\mid\mathcal F_t]}{(t+1)^2}<\infty
\qquad\text{a.s.}
\]
The martingale strong law \citep[Chapter~2]{HallHeyde1980} gives $M_h^Y(t)/t\toas0$.
Since $n_h(t)/t\to p_h>0$,
\[
\bar Y_h(t)-\mu_h=\frac{M_h^Y(t)}{n_h(t)}\toas0,
\]
and therefore $\wh T_H(t)\toas T$.

For the CLT, let $\mathbf M^Y(t)=(M_1^Y(t),\dots,M_H^Y(t))^\top$.
Because at most one stratum is sampled at each step, the predictable covariation matrix is diagonal.
We first identify its limit.
By Proposition~\ref{prop:martingale},
\[
M_h^N(t)=n_h(t)-m_0-\sum_{s=t_0}^{t-1}\pi_{s+1,h}
\]
is a martingale with bounded increments, hence $M_h^N(t)/t\toas0$ by the martingale strong law.
Combining this with $n_h(t)/t\to p_h$ gives the compensator convergence
\[
\frac{1}{t}\sum_{s=t_0}^{t-1}\pi_{s+1,h}
=
\frac{n_h(t)}{t}-\frac{m_0}{t}-\frac{M_h^N(t)}{t}
\toas p_h.
\]
Consequently,
\[
\frac{1}{t}\langle\mathbf M^Y\rangle(t)
\toas
\Sigma:=\operatorname{diag}(\sigma_1^2p_1,\dots,\sigma_H^2p_H).
\]

To verify conditional Lindeberg, fix $a=(a_1,\dots,a_H)^\top\in\mathbb R^H$ and $\varepsilon>0$.
Because at most one increment can be nonzero at each step,
\begin{align*}
&(a^\top\Delta\mathbf M^Y(s+1))^2
\1\{|a^\top\Delta\mathbf M^Y(s+1)|>\varepsilon\sqrt t\}\\
&\le
\sum_{h=1}^H a_h^2(\Delta M_h^Y(s+1))^2
\1\{|a_h|\,|\Delta M_h^Y(s+1)|>\varepsilon\sqrt t\}.
\end{align*}
For $a_h\ne0$,
\begin{align*}
&\E\Bigl[a_h^2(\Delta M_h^Y(s+1))^2
\1\{|a_h|\,|\Delta M_h^Y(s+1)|>\varepsilon\sqrt t\}
\,\bigm|\,\mathcal F_s\Bigr]\\
&\qquad=
a_h^2\pi_{s+1,h}
\E\Bigl[(Y_{h1}-\mu_h)^2
\1\{|a_h|\,|Y_{h1}-\mu_h|>\varepsilon\sqrt t\}\Bigr].
\end{align*}
The expectation on the right converges to zero by dominated convergence, since $\E[(Y_{h1}-\mu_h)^2]<\infty$.
Using $(1/t)\sum_s\pi_{s+1,h}\le1$, the conditional Lindeberg average converges to zero almost surely.
The martingale central limit theorem \citep[Chapter~3]{HallHeyde1980} yields
\[
\frac{1}{\sqrt t}\mathbf M^Y(t)
\rightsquig
\mathcal N(0,\Sigma).
\]
Finally,
\[
\sqrt t\{\wh T_H(t)-T\}
=
\sum_{h=1}^H N_h\frac{t}{n_h(t)}\frac{M_h^Y(t)}{\sqrt t}.
\]
Since $t/n_h(t)\toas1/p_h$, Slutsky's theorem gives the displayed normal limit.
\end{proof}

Corollary~\ref{cor:clt} follows from Proposition~\ref{prop:estimator-total}, Theorem~\ref{thm:consistency}, and division by $N$.

\end{document}